\documentclass[conference,letter]{IEEEtran}
\IEEEoverridecommandlockouts

\usepackage{color}

\usepackage{graphicx,tabularx,array,amsmath,amsthm,thmtools}

\usepackage{mathtools}

\usepackage{caption}
\usepackage{amsfonts}
\usepackage{bm}
\usepackage{bbm}
\usepackage{makecell}
\usepackage{multirow}
 \usepackage{amssymb}
\usepackage{txfonts}
\usepackage[T1]{fontenc}
\usepackage{tikz}
\usepackage[scr=dutchcal]{mathalfa}

\usepackage{cite}

\newtheorem{Theorem}{Theorem}
\newtheorem{Proposition}{Proposition}
\newtheorem{Lemma}{Lemma}

\newtheorem{Remark}{Remark}
\newtheorem{Definition}{Definition}

\ifCLASSINFOpdf
\else
\fi
\begin{document}
%
\title{Optimal Active Social Network De-anonymization Using Information Thresholds }

\author{\IEEEauthorblockN{ Farhad Shirani}
\IEEEauthorblockA{Department of Electrical\\and Computer Engineering\\
New York University\\
New York, New York, 11201\\
Email: fsc265@nyu.edu}
\and

\IEEEauthorblockN{Siddharth Garg}
\IEEEauthorblockA{Department of Electrical\\and Computer Engineering\\
New York University\\
New York, New York, 11201\\ 
Email: siddharth.garg@nyu.edu}

\and

\IEEEauthorblockN{Elza Erkip }
\IEEEauthorblockA{Department of Electrical\\and Computer Engineering\\
New York University\\
New York, New York, 11201\\
Email: elza@nyu.edu}
}


%


\maketitle

\begin{abstract}
In this paper, de-anonymizing internet users by actively querying their group memberships in social networks is considered. In this problem, an anonymous victim visits the attacker's website, and the attacker uses the victim's browser history to query her social media activity for the purpose of de-anonymization using the minimum number of queries.
A stochastic model of the problem is considered where the attacker has partial prior knowledge of the group membership graph and receives noisy responses to its real-time queries. The victim's identity is assumed to be chosen randomly based on a given distribution which models the users' risk of visiting the malicious website. A de-anonymization algorithm is proposed which operates based on information thresholds and its performance both in the finite and asymptotically large social network regimes is analyzed. Furthermore, a converse result is provided which proves the optimality of the proposed attack strategy.
\end{abstract}
 

%
\IEEEpeerreviewmaketitle

\section{Introduction}
Preserving user privacy is a key obstacle to the continued digitalization of the society.  Consumers increasingly view the ability to connect, communicate and collaborate over the internet without risking their personal information as an absolute necessity.
Unfortunately, this is far from the case in practice. For instance, websites track users to serve them with targeted digital advertisements. More disturbingly, web tracking can be used by state actors to stifle individuals' free speech rights, or target vulnerable minority groups. As a result, there is an urgent need to understand and quantify web users' privacy risk, that is, what is the likelihood that users on the internet can be uniquely identified using their \emph{online fingerprints}?

A user's fingerprint is the set of attributes that reflect the user's web activities: websites the user has visited and social network groups that a user is a member of ~\cite{attribute1,attribute2}, characteristics of the user's web browser (font size, for example) ~\cite{browser1}, and physical device features ~\cite{hardware1}. Fingerprinting based de-anonymization attacks build on the empirical observation that, for a large enough set of attributes, a user's fingerprints are unique.
The challenge, from an attacker's standpoint, is that an unknown user's fingerprints may not be accurately or easily available; i.e., fingerprints may be noisy and the attacker may have to actively query user devices, one attribute at a time, to measure their fingerprint. However, an attacker may only be able to issue a limited number of queries to the user's device. 

Recently, Wondracek et al.~\cite{kruegel} proposed a practical fingerprinting attack strategy which uses the user's social media group memberships for the purposes of de-anonymization. In this strategy, the attacker runs a malicious website and seeks to  de-anonymize users who visit the website. To this end, the attacker first 
uses a web gather to scrape the group memberships of users in the social network. This serves as the attacker's prior knowledge of the network which can be represented as a bipartite graph.  This is shown in the left illustration in Figure~\ref{fig:fingerprinting}. 
Note that the attacker's scanned version of the group membership graph might be different from the actual  network graph because of users' privacy settings that act as a source of noise.

When an unknown user visits the attacker's website, the attacker queries social network group memberships to find the user's identity. This is done by using browser history sniffing to ask questions of the form ``is the webpage of social network group $v$ in the user's browser history?" If yes, the attacker assumes that the user is a member of the social network 
group $v$, and if no then the attacker assumes the user is not a member of $v$. Of course, a user might be a member of a group they have not visited, or conversely, might not be a member of a group they have visited; consequently, the attacker's measurement is noisy. In this way, the attacker obtains the unknown user's partial fingerprint. This is shown in the right illustration in Figure~\ref{fig:fingerprinting}. 
By matching the partial fingerprint in the two graphs, the unknown user is de-anonymized.

\begin{figure}
\begin{center}
\includegraphics[width=0.4\textwidth]{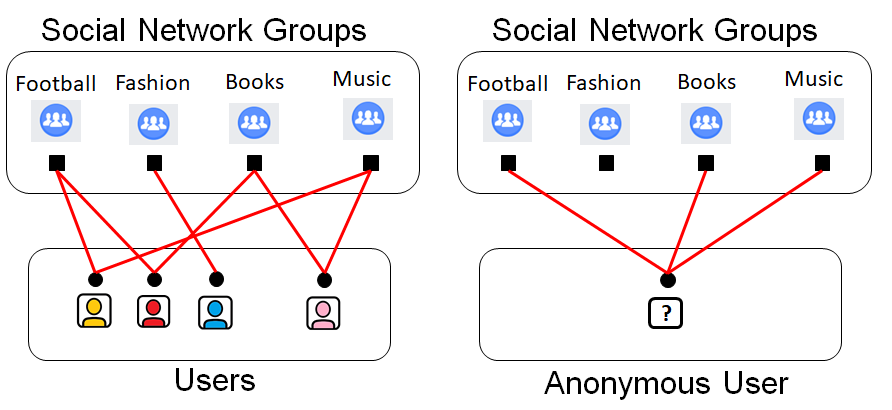}
\caption{(Left) An example of a group membership bipartite graph. (Right) A user is to be de-anonymized based on partial fingerprints.}
\label{fig:fingerprinting}
\end{center} 
\end{figure} 


Although effective, 
Wondracek et al.'s attack does not answer 
fundamental questions about the {optimal} number and type of group memberships to query, and the order in which to issue queries.
Other fingerprinting attacks proposed in literature~\cite{kruegel, finger1,finger2,finger3,finger4,finger5} have also adopted similar 
ad-hoc approaches without theoretical guarantees or analyses.

In ~\cite{Allerton}, we proposed a mathematical formulation for Wondracek et al.'s 
browser history sniffing attack~\cite{kruegel}, introduced a new strategy by making analogies to channel coding, and quantified the amount of information the attacker obtains from each query. This formulation is summarized in Section \ref{sec:form}. We showed that under the assumption that users are equally likely to visit the attacker's website, the total number of queries required for de-anonymization grows logarithmically in the number of users. Furthermore, the coefficient of the logarithm is inversely proportional to the mutual information between the random variables corresponding to the edges in the two graphs in Figure \ref{fig:fingerprinting}.

In ~\cite{Allerton}, we assumed that the victim's index is distributed uniformly over the set of all user indices. However,  internet users typically vary in terms of their activity levels on the web. More active users would be more likely to visit an attacker's website, resulting in a non-uniform distribution on the user index $J$. 
In this work, we propose an information-threshold-based de-anonymization strategy which builds upon the joint source-channel coding methods studied in ~\cite{Kostinajoint,bursh,source_channel} to devise fingerprinting attacks when the user distribution is not uniform. Roughly speaking, in the new strategy, the attacker would query the selected anonymous user's attributes sequentially and calculate the amount of information obtained, i.e. the amount of uncertainty regarding each user index based on previous query responses. The attack ends when the uncertainty is lower than a given threshold for one of the user indices. In this strategy, the
user distribution only affects the initial information values. Furthermore, we provide a tight converse which proves that the strategy is optimal in terms of expected number of queries required for de-anonymization.

The rest of the paper is organized as follows: In Section \ref{sec:form}, we provide the problem formulation. In Section III, we introduce and analyze the new attack strategy. Section IV includes the converse result for active de-anonymization attacks. Section V concludes the paper.

\section{Problem Formulation}
\label{sec:form}
In this section, we provide a rigorous formulation of the active fingerprinting problem. This is a generalization of the formulation provided in ~\cite{Allerton}. We model the group memberships in the social network by a bipartite graph. The bigraph is defined below:
\begin{Definition}
An $(n,m)$-bigraph is a structure $g_{n,m}=(\mathcal{U},\mathcal{R},{\mathcal{E}})$, where $(\mathcal{U}\bigcup \mathcal{R}, {\mathcal{E}})$ is a graph such that $|\mathcal{U}|=n$ and $|\mathcal{V}|=m$. The set $\mathcal{V}=\mathcal{U}\bigcup \mathcal{R}$ is called the vertex set and is partitioned into two subsets: 1) the user set $\mathcal{U}=\{u_1,u_2,\cdots,u_m\}$, and 2) the group set $\mathcal{R}=\{r_1,r_2,\cdots,r_n\}$. The set ${\mathcal{E}}\subseteq \{(i,j)|i\in [1,m], j\in[1,n]\}$ is called the edge set. 
\end{Definition}

\begin{Definition}
\label{def:vecs}
Let $g_{n,m}$ be an $(n,m)$-bigraph characterized by the triple $(\mathcal{U},\mathcal{R},{\mathcal{E}})$, then,
\begin{itemize}
\item{For a group $r^j, j\in [1,n]$, the set ${\mathcal{E}}_j=\{i| (i,j)\in \mathcal{E}\}, j\in [1,n]$ is called the set of members of $r_j$.}
\item{ For a user $u_i, i\in [1,m]$, the set $\mathcal{F}_i=\{j|(i,j)\in \mathcal{E}\}, i\in [1,m]$ is called the set of groups associated with $u_i$.}
\item{ The group signature of user $u_i$ is the vector $\underline{F}_i= (F_{i,1},F_{i,2},\cdots,F_{i,n})$, where
\begin{align*}
F_{i,k}=
\begin{cases}
 1 \qquad &\text{if } u_i\in \mathcal{E}_k,\\
0 &\text{otherwise}.
\end{cases}
\end{align*}}
\item{ The vector $F_{i,n_1}^{n_2}=(F_{i,n_1},F_{i,n_1+1},\cdots, F_{i,n_2})$ is called a partial group signature of $u_i$, where $1\leq n_1\leq n_2\leq n$.
}
\end{itemize}
\end{Definition}

\begin{Definition}
An $(n,m,p)$-random bigraph is a structure $g_{n,m,p}=(\mathcal{U},\mathcal{R},\pmb{\mathcal{E}})$, such that i) $g_{n,m}=(\mathcal{U},\mathcal{R},\pmb{\mathcal{E}})$ is an $(n,m)$-bigraph, and ii) The edge set $\pmb{\mathcal{E}}$ is generated randomly based on the distribution:
\begin{align*}
 P(\pmb{\mathcal{E}}=\mathcal{E})= p^{|\mathcal{E}|}(1-p)^{\left(nm- |\mathcal{E}|\right)},
\end{align*}
where $\mathcal{E} \subseteq \{(i,j)|i\in [1,m], j\in[1,n]\}$.
\end{Definition}

\begin{Definition}
A correlated pair of random bigraphs (CPRB) is a pair $\underline{g}_{n,m,P_{E_0,E_1}}=(g^0_{n,m,p_0},g^1_{n,m,p_1})$, where i) $g^{i}_{n,m,p_i}, i\in \{0,1\}$ is an $(n,m,p_i)$-random bigraph with $P(E_i=1)=p_i, i\in \{0,1\}$, and ii) The edge sets $\pmb{\mathcal{E}}^i, i\in \{0,1\}$ are generated based on the joint probability distribution: 
 \begin{align*}
 &P\left(\mathbbm{1}(e^0\in \pmb{\mathcal{E}}^0) =  \alpha, \mathbbm{1}(e^1\in \pmb{\mathcal{E}^1})=\beta)\right)=
\begin{cases}
P_{E_0,E_1}(\alpha,\beta),& \hspace{-0.1in}\text{if } e^0=e^1\\
P_{E_0}(\alpha)P_{E_1}(\beta), & \hspace{-0.1in}\text{Otherwise}
\end{cases}
\end{align*}
where $ \alpha,\beta\in \{0,1\}^2$, $e^1,e^2\in \{(i,j)|i\in [1,m], j\in[1,n]\}$, and $\mathbbm{1}(\cdot)$ is the indicator function.
\end{Definition}
The following defines an active de-anonymization problem:
\begin{Definition}
 An active de-anonymization problem is characterized by the quadruple $(\underline{g}_{n,m,P_{E_0,E_1}}, P_{J_{[m]}}, P^{UID}_{Y|Z}, P^{GM}_{Y|Z})$, where $\underline{g}_{n,m,P_{E_0,E_1}}=(g^0_{n,m,p_0},g^1_{n,m,p_1})$ is a CPRB, the random variable $J_{[m]}$ is defined on the alphabet $\mathcal{J}_{[m]}=[1,m]$, and the variables $Y$ and $Z$ are binary random variables. 
 \end{Definition}
 In the above definition,  $g^0_{n,m,p_0}$ represents the group membership graph and $g^1_{n,m,p_1}$ represents the attacker's scanned graph. The random variable $J_{[m]}$ is the index of the user $u_{J_{[m]}}$ which is to be de-anonymized. For brevity, we write $J$ instead of $J_{[m]}$ when there is no ambiguity.
 We assume a memoryless and time-invariant stochastic noise model for the responses to the attacker's queries.  The variable $Z$ represents the correct response to the attacker's query, where $Z=1$ indicates a `yes' response, and   the variable $Y$ represents the noisy response received by the attacker after making the query. As explained in ~\cite{Allerton}, the set of possible queries is divided into two categories: i) User identity (UID): a UID query asks if the unknown victim $u_J$ is user $u_j$ in the network, and ii) group Membership Queries (GM): a GM asks if the unknown victim $u_J$ is a member of the group $r_i$.  Hence, the conditional distributions $P^{UID}_{Y|Z}$ and $P^{GM}_{Y|Z}$ are the distribution of the response received by the attacker given the correct response to the attacker's UID and GM queries, respectively. Due to practical considerations explained in ~\cite{Allerton}, we assume that the responses to the UID queries are received noiselessly (i.e. $P^{UID}_{Y|Z}=I_{2}$, where $I_2$ is the unitary matrix). Loosely speaking, UID queries provide information regarding a single user index whereas GM queries provide information regarding all of the user indices in the corresponding group. As a result, it is often desirable to send GM queries when there is uncertainty among a large subset of user indices and UID queries when it is verified that the anonymous user is in a small subset of users.
 At time $t\in\mathbb{N}$, the attacker uses the vector of prior responses $Y_{1}^{t-1}$ and the scanned graph $g^1_{n,m,p_1}$ to choose the next query $x_t$. An attack strategy provides the sequence of functions $x_t(Y_{1}^{t-1},g^1_{n,m,p_1})$. If the attacker sends the UID query corresponding to the user $u_j$,  we write $x_t=u_j$, and if it sends the GM query corresponding to the group $r_i$,  we write $x_t=r_i$.
\begin{Definition}
\label{Def:rp}
 An attack strategy $\chi$ for the active de-anonymization problem $(\underline{g}_{n,m,P_{E_0,E_1}}, P_{J}, P^{UID}_{Y|Z}, P^{GM}_{Y|Z})$ is defined as a sequence of functions $x_t: g^1_{n,m,p_1}\times\{0,1\}^{(t-1)}\to\mathcal{R}\cup \mathcal{U}, t\in \mathbb{N}$. The random process $Z_t, t\in \mathbb{N}$ is defined as the sequence of correct responses:
 \begin{align*}
 Z_t=
 \begin{cases}
 1\qquad &\text{if } x_t=r_i \hspace{0.1in}\& \hspace{0.1in}J\in \mathcal{E}^0_{r_i} \text{ or }x_t=u_J,\\
 0 \qquad &\text{Otherwise.}
\end{cases}
\end{align*}
The random process $Y_t, t\in \mathbb{N}$ is the sequence of received responses. The Markov chains $Y_t\leftrightarrow Z_t \leftrightarrow x^t,Y^{t-1},Z^{t-1} ,t\in \mathbb{N}$ hold. Furthermore, 
\begin{align*}
 P_{Y_t|Z_t}=
 \begin{cases}
 P^{UID}_{Y|Z}\qquad &\text{if } x_t\in \mathcal{U}\\
P^{GM}_{Y|Z}\qquad &\text{if } x_t\in \mathcal{R}
\end{cases}
. 
\end{align*}
Finally, the random process $U_t, t\in \mathbb{N}$ is the sequence of expected responses based on the scanned graph $g^1_{n,m,p_1}$: 
 \begin{align*}
 U_t=
 \begin{cases}
 1\qquad &\text{if } x_t=r_i \hspace{0.1in}\& \hspace{0.1in}J\in \mathcal{E}^1_{r_i} \text{ or }x_t=u_J,\\
 0 \qquad &\text{Otherwise.}
\end{cases}
\end{align*}
\end{Definition}
The performance of an attack strategy is measured based on the average number of queries required for successful de-anonymization. 

\begin{Definition}
 For a given attack strategy characterized by $\chi= (x_{t}(y_1^{t-1}))_{t\in\mathbb{N}}$, the number of queries required for successful de-anonymization is defined as $Q\triangleq \min\{t| (x_t,y_t)= (u_J,1)\}$. 
 \end{Definition}
 \begin{Definition}
 For the active de-anonymization problem $(\underline{g}_{n,m,P_{E_0,E_1}}, P_{J}, P^{UID}_{Y|Z}, P^{GM}_{Y|Z})$, the minimum expected number of queries is defined as:
 \[\bar{Q}\triangleq\min_{x_t:\{0,1\}^{(t-1)}\to\mathcal{R}\cup \mathcal{U}, t\in \mathbb{N}}\mathbb{E}(Q),\]
 where the expectations is over  $\underline{g}_{n,m,P_{E_0,E_1}}$, $J$ and $Y_t,Z_t, t\in \mathbb{N}$. 
\end{Definition}

\section{The Information Threshold Strategy}
\label{sec:schemes_ITS}
In ~\cite{Allerton}, we considered the de-anonymization problem under the assumption that the victim's index $J$, is distributed uniformly over the index set $[1,m]$. Here, we relax this assumption and consider the problem when the victim's index is distributed according to an arbitrary probability distribution $P_J$ which is known a-priori to the attacker. This models the varying levels of risk averseness and risk tolerance for users over the network. More risk tolerant users are more likely to visit an attacker's website. Hence the corresponding user index is more likely to be de-anonymized. 

The information threshold strategy (ITS) is a de-anonymization strategy which operates by comparing the \emph{`information value'} of each user index with an information threshold. The information values represent the amount of uncertainty in each user index given the received query responses.
We provide bounds on the performance of the ITS for arbitrary $P^{GM}_{Y|X}$ and $P_J$, and show that the expected number of queries is proportional to the user index entropy $H(J)$ and is inversely proportional to the mutual information $I(U;Y)$ between the random variables corresponding to the received and expected responses. When the index $J$ is uniformly distributed, the first order performance (i.e. the coefficient of the logarithmic term for the expected number of queries) of the ITS is the same as that of the typical set strategy (TSS) in \cite{Allerton}. However,  
the ITS outperforms the TSS in the second order sense in the expected number of queries (i.e. the coefficient of the terms which are asymptotically smaller than the logarithmic term is improved.).
We proceed to describe the strategy.
 \begin{Definition}
For the pair of random variables $(X,Y)$ defined on the probability space $(\mathcal{X}\times\mathcal{Y}, 2^{\mathcal{X}\times\mathcal{Y}}, P_{X,Y})$, where $\mathcal{X}$ and $\mathcal{Y}$ are finite sets, the information density is defined as:
\begin{align*}
 i_{X;Y}(x;y)=\log{\frac{P_{Y|X(y|x)}}{P_{Y}(y)}}, \qquad \forall (x,y)\in \mathcal{X}\times\mathcal{Y}.
\end{align*}
\end{Definition}
Let $U_t, Y_t$ and $Z_t$ be binary random processes defined in Definition \ref{Def:rp}. Fix $\epsilon>0$ and $l\in \mathbb{N}$. Also, define the joint distribution:
\begin{align}
 P_{U,Y,Z}(u,y,z)\triangleq P_{Z}(z)P_{E_1|E_0}(u|z)P^{GM}_{Y|Z}(y|z).
 \label{eq:prob}
\end{align}

At the initial stage, the attacker forms an $m$-length vector of \emph{`information values'} corresponding to the vector of users $(u_i)_{i\in [1,m]}$. The initial information value for user $j\in [1,m]$ is defined as: 
\begin{align*}
 I_{0}(j)\triangleq \log{\frac{1}{P_{J}(j)}}.
\end{align*}
Once the vector of information values is formed, the attack progresses in $l$ steps.  In each step, the attacker sends group membership queries sequentially starting with the query corresponding to the first group $r_1$, and proceeding by increasing the group index. Hence, it makes queries regarding the values of $Z^{\tau^*}=(F_{J,1}, F_{J,2}, \cdots,F_{J,\tau^*})$ defined in Definition \ref{def:vecs}, where $\tau^*$ is defined later. The attacker receives the noisy sequence of responses $(Y_1,Y_2,\cdots, Y_{\tau^*})$. The information value query number $n'$ is defined below:
\begin{align*}
 I_{n'}(j)=i_{U^{n'};Y^{n'}}(U^{n'},Y^{n'})-I_0(j), j\in [1,m], n'\in [1,\tau^*],
\end{align*}
GM queries are concluded if there exists at least one index $j$ for which the information value $I_{\tau^*}(j)$ is larger than the information threshold $\log{\frac{1}{\epsilon}}$. More precisely, the attacker defines
\begin{align*}
 \tau_j=\min\{n'\geq 0: I_{n'}(j)\geq \log{\frac{1}{\epsilon}}\},j \in[1,m], 
\end{align*}
and the stopping time
$ \tau^*=\min_{j\in [1,m]} \tau_j$. The attacker then finds the user with the maximum information value:
\begin{align*}
\widehat{J}=argmax_{j\in [1,m]} I_{\tau^*}(j).
\end{align*}
Next, the attacker sends a  UID query to verify that the user index is equal to $\widehat{J}$. In other words, it transmits $x_{\tau^*+1}=u_{\widehat{J}}$. The algorithms ends if the output $y_{\tau^*+1}=1$  is received. If the attacker fails to recover $J$ in this step, it proceeds to the next step. In summary, the attack strategy in the first step is given below:
 \begin{align*}
 x_t= 
 \begin{cases}
 r_t\qquad &\text{ if } t\leq \tau^*,\\
 u_{i_{\widehat{J}}}& \text{ if } t=\tau^*+1.
\end{cases}
\end{align*}
If the UID returns a negative response, in the next step, the attacker resets the information values to their initial values. It repeats the previous step for the next set of group indices. If the attack strategy reaches the $l$th step, all of the possible UID's are sent for all remaining users until the user is de-anonymized. The number of steps $`l'$ is chosen such that the probability of reaching the $l$th step is small enough.

We denote the number of queries in this attack strategy by $Q_{ITS}$. The following theorem provides bounds on the expected number of queries: 
\begin{Theorem}
For the ITS strategy:
 \begin{align}
\mathbb{E}(Q_{ITS})\leq \frac{1}{(1-\epsilon)}\left(\frac{H(J)+\log\frac{1}{\epsilon}+i_{max}}{I(U;Y)}+1\right)+\frac{m}{2}\epsilon^l ,
\end{align}
where $i_{max}\triangleq \max_{u,y} i_{U;Y}(u;y)$, and $P_{U,Y}$ is defined in equation \eqref{eq:prob}, provided that the number of groups `$n$' satisfies 
 \begin{align}
 \frac{1}{(1-\epsilon)\epsilon^l}\left(\frac{H(J)+\log\frac{1}{\epsilon}+i_{max}}{I(U;Y)}+1\right)<n.
 \label{eq:g1}
\end{align}
Particularly, if $\epsilon= \frac{\log\log{m}}{\log{m}}$, and $l\triangleq \frac{\log{m}}{\log\log{m}-\log\log\log{m}}$, then the inequality 
\begin{align}
 \mathbb{E}(Q_{ITS})\leq  \frac{H(J_{[m]})}{I(U;Y)}+O(\log\log{m})
\label{eq:bound4}
\end{align}
 holds, provided that
\begin{align}
   \frac{H(J_{[m]})}{I(U;Y)}\log\log{m}<n.
   \label{eq:g2}
\end{align}

\label{thm:4}
\end{Theorem}
\begin{proof}
 Please refer to the appendix.
\end{proof}
\begin{Remark}
In Theorem \ref{thm:4}, we have assumed that the number of groups $n$ satisfies: 
\begin{align*}
 \frac{1}{(1-\epsilon)}\left(\frac{H(J)+\log\frac{1}{\epsilon}+i_{max}}{I(U;Y)}+1\right)+\frac{m}{2}\epsilon^l<n.
\end{align*}
\end{Remark}

\section{Converse Theorems for Active De-anonymization}
In this section, we prove that the bound provided in Theorem \ref{thm:4} is tight in the first order sense (i.e. the coefficient of $H(J)$ in Equation \eqref{eq:bound4} cannot be improved.). To this end, we build upon the converse of Burnashev's reliability function for point-to-point communication over channels with noiseless feedback ~\cite{bursh,simple,Kostinajoint}. More specifically, for a given active de-anonymization strategy, we introduce a dual code for communication over a channel with feedback, and use the converse results in ~\cite{Kostinajoint,bursh,simple} to provide bounds on the optimal performance. In our analogy, the victim's index in the fingerprinting problem is analogous to the message which is to be sent over the channel with feedback and the queries are analogous to the channel inputs.
The following provides the standard definition for a variable length code for a discrete memoryless channel with feedback.

\begin{Definition}
 A $(T,\epsilon)$-variable length code for the channel with feedback $(\mathcal{U},\mathcal{Y},P_{Y|U})$ 
 transmitting message $J$ (with alphabet $\mathcal{J}$) is defined by: i) the common information $V\in \mathcal{V}$ shared between the encoder and decoder prior to the start of the communication, ii) a sequence of encoding functions $U_n: \mathcal{V}\times\mathcal{J}\times\mathcal{Y}^{n-1}\mapsto \mathcal{U}, n\in\mathbb{N}$, iii) a sequence of decoding functions $\widehat{J}_n:\mathcal{V}\times\mathcal{Y}^n\mapsto \mathcal{J}$, and iv) a non-negative integer-valued random variable $\tau$.
 
The reconstruction of the message at the decoder is given by $\widehat{J}= \widehat{J}_{\tau}(V,Y^{\tau})$. The average transmission time for a given code is defined as $T=\mathbb{E}(\tau)$. The probability of error for the code is $\epsilon=P(\widehat{J}\neq J)$.
\end{Definition}
We use the following result:
\begin{Lemma}\cite{Kostinajoint}
 The average transmission time for a family of $(T_n,\epsilon_n)$-variable length block codes for the channel with noiseless feedback $(\mathcal{U},\mathcal{Y},P_{Y|U})$ is lower-bounded by:
 \begin{align*}
T_n\geq \frac{H(J)}{I(U;Y)}+O(\epsilon_n), 
\end{align*}
 where $\epsilon_n\to 0$ as $n\to \infty$.
 \label{lem:con}
\end{Lemma}
The following provides a converse for the performance of active de-anonymization strategies:
\begin{Theorem}
Let $(\underline{g}_{n_t,m_t,P_{E_0,E_1}}, P_{J_{[m_t]}}, P^{UID}_{Y|Z}, P^{GM}_{Y|Z}), t\in \mathbb{N}$ be a sequence of active de-anonymization problems where $m_t\to \infty$ as $t\to \infty$ and $n_t$ is an arbitrary sequence. Consider the sequence of attack strategies $\chi^t, t\in \mathbb{N}$. Define $\bar{Q}_{\chi^t}$ as the expected number of queries for successful de-anonymization for the strategy $\chi^t$. Then,
 \begin{align*}
 \bar{Q}_{\chi^t}\geq \frac{H(J_{[m_t]})}{I(U;Y)}+ O({\log{m_t}}).
\end{align*}
\label{th:con}
 \end{Theorem}
 \begin{proof}
 Please refer to the Appendix.

\end{proof}

\section{Conclusion}
We have studied the active de-anonymization problem for general non-equiprobable user indices. We have introduced the ITS de-anonymization strategy which operates based on information thresholds. The new strategy measures the amount of uncertainty in the user indices given the received query responses. We have characterized the performance of the ITS both for social networks with a fixed, finite number of users as well as for asymptotically large social networks. Finally, we have provided a converse which shows the first-order optimality of the proposed approach.  
 
\appendix
\subsection{Proof of Theorem \ref{thm:4}}
The proof builds upon the method in \cite{Kostinajoint}.
Fix $\epsilon>0$ and $l\in \mathbb{N}$. Let $\tau_j$, $\tau^*$, and $\widehat{J}$ be defined as in Section \ref{sec:schemes_ITS}. Let $E_i, i\in [1,l-1]$ be the event that $\widehat{J}$ is equal to the victim's index $J$ in the $i$th step. We have the following:
\begin{align}
\nonumber
&  \mathbb{E}\left(Q_{ITS}\right){=}
\\&\nonumber
   P\left(E_1\right)\mathbb{E}\left(Q_{ITS}|E_1\right)+
   P\left(E_1^c\bigcap E_2\right)\mathbb{E}\left(Q_{ITS}|E^c_1\bigcap E_2\right)
+\cdots
\\&\nonumber
 +P\left(\bigcap_{i=1}^{l-2} E_i^c\bigcap E_{l-1}\right)\mathbb{E}\left(Q_{ITS}|\bigcap_{i=1}^{l-2} E_i^c\bigcap E_{l-1}\right)
+  P\left(\bigcap_{i=1}^l E_i^c\right)\cdot \frac{m}{2}\\
&\nonumber
 = P\left(E_1\right)\left(\mathbb{E}(\tau^*)+1\right)+P\left(E_1^c\bigcap E_2\right)\left(2\mathbb{E}(\tau^*)+2\right)+\cdots
\\&\nonumber
  +P\left(\bigcap_{i=1}^{l-2} E_i^c\bigcap E_{l-1}\right)\left(\left(l-1\right)\mathbb{E}(\tau^*)+l-1\right)+  P\left(\bigcap_{i=1}^l E_i^c\right)\cdot\frac{m}{2}\\
&\nonumber
 \stackrel{(a)}{=} \sum_{i=1}^{l-1}i\left(1-\epsilon\right)\epsilon^{i-1}\left(\mathbb{E}\left(\tau^*\right)+1\right)+\frac{m}{2}\epsilon^{l-1}\\
&= \frac{1}{\left(1-\epsilon\right)}\left(\mathbb{E}\left(\tau^*\right)+1\right)+\frac{m}{2}\epsilon^{l-1}, \label{eq:ITS1}
\end{align}
where (a) follows from the fact that the events $E_i$ are mutually independent and equiprobable with $P(E_i)=1-\epsilon, i\in [1,l-1]$. 
Next, we derive an upper-bound for $\mathbb{E}(\tau^*)$.
\begin{Proposition}
 The following bound on the expectation of $\tau^*$ holds:
\begin{align*}
  \mathbb{E}(\tau^*)\leq \frac{H(J)+\log\frac{1}{\epsilon}+i_{max}}{I(U;Y)}.
\end{align*}
\end{Proposition}

\begin{proof}
First, note that $\mathbb{E}(\tau^*)\leq \mathbb{E}(\tau_J)$ by definition of $\tau^*$. So, it is enough to prove the upper bound on $\mathbb{E}(\tau_J)$. Fix $j\in \mathbb{N}$. Let $t= \min\{\tau_J,j\}$. Note that:
 \begin{align*}
 \mathbb{E}\left(I_{t}\left(J\right)\right)&= \mathbb{E}\left(i_{U^t;Y^t}\left(U^t;Y^t\right)\right)-I_o(j)\\
 &\stackrel{(a)}{=}\mathbb{E}\left(\sum_{j=1}^ti_{U;Y}\left(U_j;Y_j\right)\right)-H\left(J\right)\\
 &\stackrel{(b)}{=}I(U;Y)\mathbb{E}(t)-H(J),
\end{align*}
where in (a) we have used the memoryless property of the network and (b) follows from the smoothing property of expectation and the fact that  $i_{U;Y}\left(U;Y\right)$ is constant in $j$. Also, note that $I_t(J)$ is an increasing function of $t$ and at each step the increment is less than or equal to $i_{max}$. It follows that:
\begin{align*}
 \mathbb{E}\left(I_{t}\left(J\right)\right)\leq \mathbb{E}\left(I_{t-1}\left(J\right)\right)+i_{max}\leq \log\frac{1}{\epsilon}+i_{max},
\end{align*}
where we have used the fact that $i_{U;Y}(u,y)\leq i_{max}, \forall u,y\in \mathcal{U}\times\mathcal{Y}$ and that by the definition of $\tau_J$, we have $I_{t-1}(J)\leq \log\frac{1}{\epsilon}$ since $t-1<\tau_J$. So, far we have shown that $I(U;Y)\mathbb{E}(t)-H(J)\leq \log\frac{1}{\epsilon}+i_{max}, \forall j\in \mathbb{N}$. By the monotone convergence theorem it follows that $I(U;Y)\mathbb{E}(\tau_J)-H(J)\leq \log\frac{1}{\epsilon}+i_{max}$. This completes the proof of the proposition.
\end{proof}
Applying the above proposition to Equation \eqref{eq:ITS1} gives:
\begin{align*}
   \mathbb{E}\left(Q_{ITS}\right)\leq \frac{1}{\left(1-\epsilon\right)}\left(\frac{H\left(J\right)+\log\frac{1}{\epsilon}+i_{max}}{I\left(U;Y\right)}+1\right)+\frac{m}{2}\epsilon^{l-1}.
  \end{align*}
Take $\epsilon= \frac{\log\log{m}}{\log{m}}$, and $l= \frac{\log{m}}{\log\log{m}-\log\log\log{m}}$, then \[\mathbb{E}\left(Q_{ITS}\right)\leq  \frac{H\left(J\right)}{I\left(U;Y\right)}+O\left(\log\log{m}\right).\]
The bounds in \eqref{eq:g1} and \eqref{eq:g2} follow from the fact that from the Markov inequality, $P(\tau^*>\frac{\mathbb{E}(\tau^*)}{\epsilon^l})\leq \epsilon^l$.
This completes the proof of Theorem \ref{thm:4}.

\subsection{Proof of Theorem \ref{th:con}}
We provide an outline of the proof. For the given attack strategy $\chi^n= (X_t)_{t\in \mathbb{N}}$, we define the dual code for communication over the channel with feedback $(\mathcal{U},\mathcal{Y},P_{Y|U})$ which transmits the message $J_{[m]}$ as follows. 
At time $t$, assume that
 the sequence $Y^{t-1}$ is noiselessly available at the transmitter through feedback. The encoder first calculates the query $X_t(Y^{t-1})$ given the the attack strategy $\chi^n$ when the attacker is to de-anonymize the user index $J_{[m]}$. If $X_t$ is a  GM query corresponding to the group $r_t$, then $U_t=\mathbbm{1}(u_{J_{[m]}}\in r_t)$ is transmitted as the channel input. If $X_t$ is a UID query corresponding to the user $u_{\widehat{J}_{[m]}}$, then the transmitter uses a repetition code as described next. If $J_{[m]}=\widehat{J}_{[m]}$, the all-ones sequence of length $t_{rep}$ is transmitted, and if $J_{[m]}\neq \widehat{J}_{[m]}$, the all-zeros sequence of length $t_{rep}$ is transmitted, where $t_{rep}$ will be determined later. The decoder uses a majority logic decoder for the repetition code to determine whether the response to the UID query is positive or not. By construction, the probability of error of this sequence of codes $\epsilon_m$ is proportional to the probability of error of the repetition code since if the response to the UID query is positive and correctly decoded, then the transmission is succsessful. 
Hence, the probability of error approaches 0 as $m\to \infty$ if $t_{rep}=\omega(\log{(\frac{1}{\epsilon_m}}))$. Let $T_m$ be the expected transmission time for this code. The statement of theorem follows from Lemma \ref{lem:con} by taking $\epsilon_{m}=\omega(\frac{1}{m})$:
\begin{align*}
 \bar{Q}_{\chi^n}\geq T_m+\Theta(t_{rep})\geq \frac{H(J)}{I(U;Y)}+ O({\log{m}}).
\end{align*}

%
%

\bibliographystyle{unsrt}

\end{document}